\documentclass[10pt]{IEEEtran}
\usepackage{amsmath,amsfonts, amsthm, amssymb}
\usepackage{algorithmic}
\usepackage{algorithm}
\usepackage{array}
\usepackage[caption=false,font=normalsize,labelfont=sf,textfont=sf]{subfig}
\usepackage{textcomp}
\usepackage{stfloats}
\usepackage{url}
\usepackage{verbatim}
\usepackage{graphicx}
\usepackage{cite}
\usepackage{caption}
\usepackage[letterpaper, top=0.7in, bottom=1in, textwidth=7.26in]{geometry}
\usepackage{xcolor}
\usepackage{enumitem}
\usepackage{thmtools}
\usepackage{centernot}
\usepackage{nopageno}

\declaretheoremstyle[
  headfont=\bfseries, 
  bodyfont=\normalfont,
]{boldstyle}

\declaretheorem[style=boldstyle, name=Theorem]{theorem}

\declaretheorem[style=boldstyle, name=Definition]{definition}

\newcommand{\WW}{\mathbf{W}} 
\newcommand{\WF}{\mathbf{W}_F} 
\newcommand{\SF}{\mathbf{S}_F} 
\newcommand{\WI}{\mathbf{W}_I} 
\newcommand{\Entity}{\mathbf{E}} 
\newcommand{\Entitya}{a} 
\newcommand{\Entityb}{b} 
\newcommand{\Predicate}{\mathbf{P}} 
\newcommand{\PredicateR}{R} 
\newcommand{\NumState}{n} 
\newcommand{\State}{\mathbf{s}} 
\newcommand{\Width}{w} 
\newcommand{\Language}{\mathcal{L}} 
\newcommand{\IndexSetPred}{\mathcal{I}_P} 
\newcommand{\IndexSetQSent}{\mathcal{I}_\mathcal{Q}} 
\newcommand{\IndexSetState}{\mathcal{I}_{\SF}} 
\newcommand{\IndexSetConst}{\mathcal{I}_{\mathcal{C}}} 
\newcommand{\ProbMeas}{\mathcal{P}_{I}} 
\newcommand{\ProbLoc}[2]{\mathcal{P}^{#1}_{I, #2}} 

\newcommand{\argmin}{\mathop{\mathrm{arg\,min}}} 
\newcommand{\argmax}{\mathop{\mathrm{arg\,max}}} 

\begin{document}

\title{DISCD: Distributed Lossy Semantic Communication for Logical Deduction of Hypothesis}

\author{Ahmet~Faruk~Saz,~\IEEEmembership{Student Member,~IEEE,}
        Siheng~Xiong,~\IEEEmembership{Student Member,~IEEE,}
        and~Faramarz~Fekri,~\IEEEmembership{Fellow,~IEEE}
}

\maketitle

\begin{abstract}

In this paper, we address hypothesis testing in a distributed network of nodes, where each node has only partial information about the State of the World (SotW) and is tasked with determining which hypothesis, among a given set, is most supported by the data available within the node. However, due to each node's limited perspective of the SotW, individual nodes cannot reliably determine the most supported hypothesis independently. To overcome this limitation, nodes must exchange information via an intermediate server. Our objective is to introduce a novel distributed lossy semantic communication framework designed to minimize each node’s uncertainty about the SotW while operating under limited communication budget. In each communication round, nodes determine the most content-informative message to send to the server. The server aggregates incoming messages from all nodes, updates its view of the SotW, and transmits back the most semantically informative message. We demonstrate that transmitting semantically most informative messages enables convergence toward the true distribution over the state space, improving deductive reasoning performance under communication constraints. For experimental evaluation, we construct a dataset designed for logical deduction of hypotheses and compare our approach against random message selection. Results validate the effectiveness of our semantic communication framework, showing significant improvements in nodes’ understanding of the SotW for hypothesis testing, with reduced communication overhead.

\end{abstract}


\section{Introduction}

Low latency, minimal power consumption, and efficient bandwidth are critical in today’s hyper-connected systems. Autonomous technologies like self-driving cars and drones depend on real-time data exchange to ensure safety, while smart cities require reliable communication to optimize services. In healthcare, semantic communication enhances data accuracy for remote monitoring, improving patient outcomes. In industry 4.0, it ensures efficient, real-time information flow between machines and sensors, driving automation and advancing manufacturing intelligence. Semantic systems will support fast and explainable decision-making in next-generation applications such as automated truth verification, data connectivity and retrieval, financial modeling, risk analysis, privacy-enhancing technologies, and optimized agriculture \cite{Gund_Yen, Park2022, Hashash}. In these applications, distributed systems may need to exchange data for deduction of a hypothesis  to improve their generalization ability, decision-making performance, or statistical analysis. However, due to privacy concerns and communication limitations, transmitting their entire dataset may not be feasible. In such a case, semantic communication may play an essential role.

Rudolf Carnap's foundational work introduced a model-theoretic approach to semantic communication compatible with First-Order Logic (FOL) to represent states. His framework provided a way to quantify semantic information using inductive probability distributions. Later, Hintikka expanded on Carnap's work by addressing infinite universes and generalization issues \cite{c10}. Early post-5G semantic communication methods like D-JSCC and autoencoders employed neural network-based solutions \cite{yashas, GLi, Bennis2022b, gul-yen, Qin2021SemanticCP, aguerri2019distributed, zaidi2020distributed, xu2020acceleration, gupta2020training, krouka2021communication, Xie2020DeepLE, Shao2022ATO, e26020102, e24060846}. Traditional model-theoretic approaches like Carnap's and Floridi's \cite{floridi} offer interpretability but suffer from computational inefficiency and lack probabilistic foundations. Newer approaches include information expansion and knowledge collusion \cite{e24121842}, knowledge-graph integration enhanced with probability graphs \cite{e26050394}, memory supported deep learning based approaches \cite{Xie2023SemanticCW} and topological approaches \cite{Zhao2022SemanticNativeCA}.  These methods, while effective, are not particularly tailored to FOL reasoning tasks \cite{xiongtilp, xiong2024teilp, yang2024temporal, yang2023harnessing, xiong2024large} such as deduction of a hypotheses.

In this paper, we propose a novel lossy semantic communication framework for a distributed set of nodes, each with a deductive task and partial, possibly overlapping views of the world aiming to identify, among a set of hypotheses, the one best supported by the SotW. Nodes communicate with a central server to enhance their understanding of the SotW, transmitting the most informative messages about their local environments under bandwidth limitations and privacy concerns. The server aggregates data and returns the most informative messages to help nodes refine their local probability distributions over the state space. 

Our semantic encoding algorithm selects messages based on the semantic informativeness of a message to a specific observer within communication constraints. Messages that have higher degree of confirmation (e.g., inductive logical probability) are deemed more informative as they leave less uncertainty about the SotW. In addition, this new paradigm is able to effectively manage large state spaces, maximizing each node's understanding of the SotW while enabling more robust deduction performance despite nodes being unaware of each other's or the central node's full perspective and limited communication resources.  

Finally, we showcase that as the number of communication rounds increases, nodes distribution over the state space converge towards the true distribution, minimizing Bayes risk over hypothesis set. Empirical validation through experiments using a FOL deduction dataset shows that our framework significantly improves nodes' understanding of the world and deduction performance with reduced overhead. 

\section{Background on FOL, Inductive Probabilities, and Semantic Information}

\subsection{Representing the State of the World via First-Order Logic}

We begin by illustrating how the state of the world can be represented using First-Order Logic (FOL). Consider a finite world $ \WF $ characterized by a finite set of entities $\Entity $ and a finite set of predicates $\Predicate$. The set of all possible states of $\WF$ is denoted by \(\SF  = \{ \State_1, \State_2, \dots, \State_\NumState \} \), where \( \NumState = 2^{|\Predicate| \times |\Entity|^2} \). Each state \( \State_i \) provides a complete description of the world using the logical language \( \Language  \), which includes the predicates \( \Predicate \), entities \( \Entity \), quantifiers \( \forall, \exists \), and logical connectives \( \land, \lor, \neg \).

\begin{definition}[State Description]
A \textit{state description} \( \State_i \) in \( \WF \) is defined as:
\begin{equation}
    \State_i = \bigwedge_{r \in \IndexSetPred} \bigwedge_{\Entitya, \Entityb \in \Entity} \delta_{r} \PredicateR_r(\Entitya, \Entityb),
\end{equation}
where \( \delta_{r} \in \{+1, -1\} \) indicates the presence (\( +1 \)) or negation (\( -1 \)) of the predicate \( \PredicateR_r \) for entities \( \Entitya \) and \( \Entityb \), $\IndexSetPred$ is the index set of predicates.
\end{definition}

In this setup, each state is a conjunction (\( \bigwedge \)) of a specific enumeration of all possible predicates (or their negations), \( \PredicateR_r \text{ for } r \in \IndexSetPred\), applied to all ordered pairs of entities, \( (\Entitya, \Entityb) \in \Entity\). Any FOL sentence \( m \) in this finite world can thus be expressed as a disjunction of these state descriptions.

\begin{theorem}[FOL Representation in Finite World  \cite{c2}]
Any FOL sentence \( m \) in \( \WF \) can be represented as:
\begin{equation}
    m \equiv \bigvee_{\substack{i \in \IndexSetState \\ \text{such that } m \models \State_i}} \State_i,
\end{equation}
where $\models$ is the logical entailment operator and $\IndexSetState$ is the index set of finite states.
\end{theorem}

When extending to an infinite world \( \WI \) where entities \( \Entity \) becomes a countably infinite set, a more intricate approach is needed. We introduce the concept of \textit{Q-sentences} to handle this complexity.

\begin{definition}[Q-sentence]
A \textit{Q-sentence} \( Q_i(\Entitya, \Entityb) \) in \( \WI \) is defined as:
\begin{equation}
    Q_i(\Entitya, \Entityb) = \bigwedge_{r \in \IndexSetPred} \bigwedge_{\Entitya, \Entityb \in \Entity} \delta_{r} R_r(\Entitya, \Entityb) \land \delta'_{r} R_r(\Entityb, \Entitya),
\end{equation}
where \( \delta_{r}, \delta'_{r} \in \{+1, -1\} \) indicate the presence or negation of predicates. This formulation encompasses all combinations, resulting in \( 4^{|\Predicate| \times |\Entity|^2} / 2\) distinct Q-sentences, i.e., $i = 1, ...,  4^{|\Predicate| \times |\Entity|^2} / 2$.
\end{definition}

Building upon Q-sentences, we define \textit{attributive constituents} to capture the relationships of individual entities.

\begin{definition}[Attributive Constituent]
An \textit{attributive constituent} \( C_t(x) \) in \( \WI \) is:
\begin{equation}
    C_t(x) = \bigwedge_{i \in \IndexSetQSent} \delta_{k} \{(Q_i(x, b)): \exists b\},
\end{equation}
where \( \delta_{i} \in \{+1, -1\} \), $\IndexSetQSent$ is the index set over Q-sentences, and $i = 1, ... |\IndexSetQSent|$, and $\Entityb \in \Entity$. The expression encapsulates all relationships entity \( x \) has within the world, expressing \textit{what kind of individual} $x$ is.
\end{definition}

To represent the entire infinite world, we define \textit{constituents} as combinations of attributive constituents.

\begin{definition}[Constituent]
A \textit{constituent} \( C^\Width \) of width \( \Width \) in \( \WI \) is a conjunction of attributive constituents:
\begin{equation}
    C^\Width = \bigwedge_{t=1}^{\Width} \; \{ C_t(x): \exists x\},
\end{equation}
which represents the relational structure of entities in the world \( \WI \). Here, the width \( \Width \) is the number of attributive constituents \( C_t(x) \) present in the conjunction forming \( C^\Width \).
\end{definition}

Any FOL sentence \( m \) in \( \WI \) can then be expressed as a disjunction of such constituents.

\begin{theorem}[FOL Representation in Infinite World \cite{HINTIKKA196548}]
Any FOL sentence \( m \) in \( \WI \) can be represented as:
\begin{equation}\label{const_dis}
    m \equiv \bigvee_{\substack{j \in \IndexSetConst \\ \text{such that } m \models C^j}} C^j,
\end{equation}
where $\IndexSetConst$ is the index set of constituents.
\end{theorem}

As an illustration, consider the statement "Every person owns a book.", or more explicitly, "For every x, if x is a person, then there exists a y such that y is a book and x owns y.", formally expressed as:
\begin{equation}
    m = \forall x \left( Person(x) \rightarrow \exists y \left( Book(y) \land Owns(x, y) \right) \right).
\end{equation}

\noindent To represent \( m \) as a disjunction of constituents, we identify all constituents \( C^{j'} \) where \( m \) holds. In each such constituent, for every entity \( x \) satisfying \( Person(x) \), the attributive constituent \( C'_t(x) \) must include:
\[
    C'_t(x) = 
    \begin{cases}
        \left[ Person(x) \land \exists y \left( Book(y) \land Owns(x, y) \right) \right] \land \\
        \left[ \bigwedge_{i \in \IndexSetQSent} \delta_{i} \{(Q_i(x, b)): \exists b\} \right],
        \textbf{ if } Person(x) \equiv 1, \\
        \\
         \left[  \lnot Person(x) \right] \land \bigwedge_{i \in \IndexSetQSent} \delta_{i}  \{(Q_i(x, b)): \exists b\} , otw.
    \end{cases}
\]

\noindent Each constituent \( C^{j'} \) is then constructed as:

\begin{equation}
    C^{j'} = \bigwedge_{t=1}^{j'} \{ C_t'(x): \exists x\},
\end{equation}
ensuring that all entities are accounted for. Therefore, \( m \) can be expressed as:
\begin{equation}
    m \equiv \bigvee_{\substack{j' \in \IndexSetConst \\ \text{such that } m \models C^{j'}}} C^{j'},
\end{equation}

State descriptions and constituents act like basis vectors in a vector space (as they partition the space and hence any two are mutually exclusive), allowing us to define a probability measure over them. The inductive logical probability of any sentence in the FOL language \( \mathcal{L} \) is then the sum of the probabilities of its constituents as per axioms of probability. In the next section, we will discuss how to define this probability measure using the frameworks proposed by Carnap and Hintikka.

\subsection{Inductive Logical Distribution on States}

In FOL-based world \( \WW \text{ (could be } \WF \text{ or } \WI) \), we introduce a probability measure \( \ProbMeas \) over the fundamental elements—either state descriptions \( \State_i \) or constituents \( C^j \), collectively referred to as "world states" \( S_i \). Let $S_i \in \mathcal{S}$ where $\mathcal{S}$ is the state space. To define this measure, we make use of \textit{inductive logical probabilities}, which serve as inductive Bayesian posteriors, built upon empirical observations and a prior distribution. 

\begin{definition}[Inductive Logical Probability]
For any state \( S_i \) and corresponding evidence \( e \) (observations) in \( \WW \), the inductive logical probability (a.k.a., degree of confirmation) \( c(S_i, e) \) that evidence \( e \) supports sentence \( S_i \) is defined as:
\begin{equation}\label{c_func}
c(S_i, e) = p(S_i \mid e) = \frac{p(S_i \land e)}{p(e)} = \frac{l_{S_i} + \lambda(w_{S_i})/w_{S_i}}{l + \lambda(w_{S_i})},
\end{equation}
where \( l \) is the number of observations, \( l_{S_i} \) is the number of observations confirming state \( S_i \), \( w_{S_i} \) is the weight assigned to state \( S_i \), \( \lambda(w_{S_i}) \) is the prior coefficient, dependent on \( w_{S_i} \). The parameter \( \lambda(w_{S_i'}) \), ranging from 0 to \( \infty \), balances the weight of prior against empirical data. 
\end{definition}

Then, the inductive logical probability measure \( \ProbMeas \) over $\mathcal{S}$ is defined as $\ProbMeas = \{c(S_i, e) \mid i \in \mathcal{I}_{\mathcal{S}}\}$ where $\sum_{i \in \mathcal{I}_{\mathcal{S}}} c(S_i, e) = 1$. In the following section, we make use of inductive logical probabilities to define a semantic information metric to effectively quantify semantic informativeness regarding states of the world.

\subsection{Measuring Semantic Uncertainty about the SotW}

To quantify the remaining uncertainty after new observations, Carnap proposed the concept of "cont-information". This metric assesses how observations influence an observer's understanding of the SotW by evaluating the extent to which they decrease observer's uncertainty.

\begin{definition}
The \textit{cont-information} in a world \( \WW \) quantifies the informativeness of an observation \( e \) to a specific observer by measuring the remaining uncertainty about the SotW.  Given a state \( S_i \) and evidence \( e \) (observations) in \( \WW \), cont-information is defined as:
\begin{equation}
\text{cont}(S_i; e) = 1 - c(S_i, e) = 1 - p(S_i | e),
\end{equation}
\end{definition}

In simple terms $\text{cont}(S_i; e)$ measures the uncertainty remaining about state $S_i$ after observing e. Next, the cont-information measure is used to devise a communication framework. This measure captures the dynamic and cumulative nature of evidence accumulation in the communication and, as a result, the refinement in an observer’s understanding of the world state.
 
\section{Distributed Semantic Communication Framework}

\begin{figure}[!t]
\centering
\includegraphics[width=0.4\textwidth]{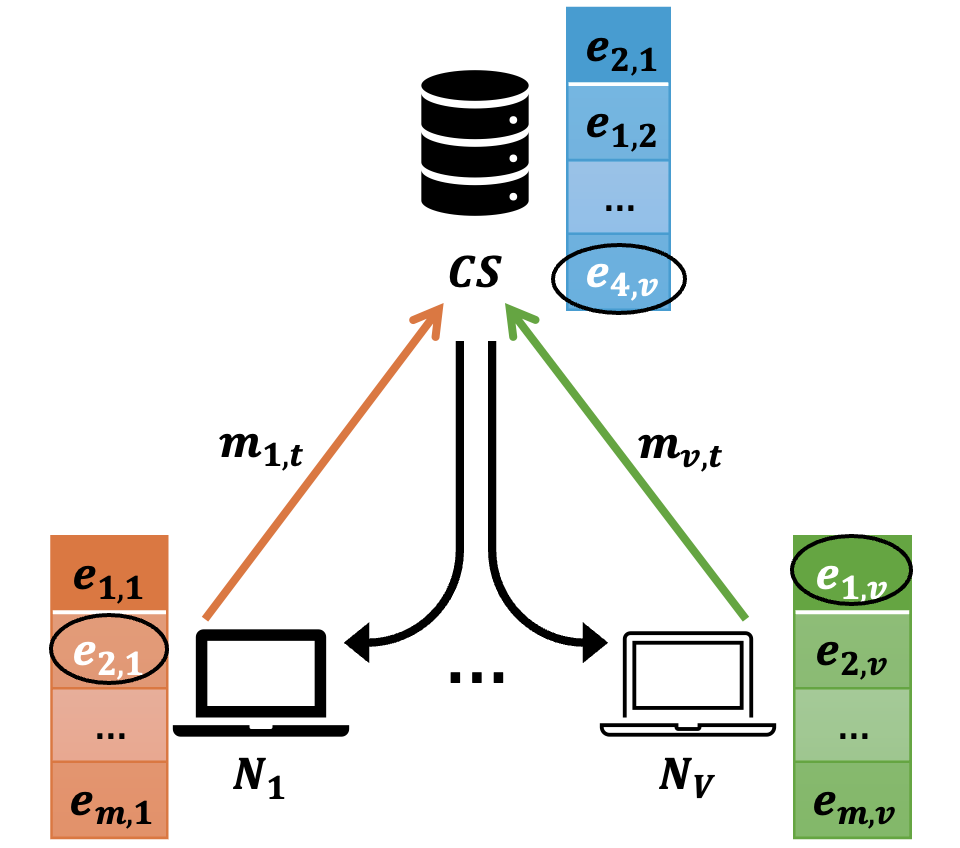}
\captionsetup{justification=centering}
\caption{Distributed semantic communication framework. Nodes exchange most semantically informative messages.}
\label{fig:my}
\end{figure}

\subsection{Hypothesis Deduction Problem}

We start by formalizing a hypothesis deduction problem at each user node \( N_j \) in the wireless network. Let \( \mathcal{H}_j = \{ h_i \mid i \in \mathcal{I}_\mathcal{H} \} \) be a finite or infinite set of self-consistent statements representing possible hypotheses about the world \( \WW \) at node $N_j$. There exists a hypothesis \( h^* \in \mathcal{H} \)  which is most consistent, among all other hypothesis, with the actual SotW. The statements \( h_i \) can be represented by a non-empty disjunction of state descriptions or constituents, as per Theorem 1 and 2. At each node \( N_j \), the task is to deduce the hypothesis $h_j^*$ (the most consistent hypothesis with the SotW). However, as each node has incomplete information, this deduction cannot be concluded reliably. The aim of the communication is, then, to receive most cont-informative information regarding world state so as more accurately determine which hypothesis is supported the most by world (entity) population.

\subsection{System Model}

As shown in Fig.1, we consider a distributed network comprising a set of edge nodes \( \mathcal{N} = \{ N_1, N_2, \dots, N_v \} \) and a central server \( CS \). Each edge node possesses evidence \( e_j \) obtained through observations of the world \( \WW \), which can be finite (\( \WF \)) or infinite (\( \WI \)). Each node \( N_j \) maintains a local inductive probability distribution \( \ProbLoc{j}{t} \) over the state space \( \mathcal{S} \), reflecting its view of the State of the World (SotW) at round \( t \). Each node \( N_j \) aims to solve a hypothesis deduction problem \( \mathcal{B}_j \), involving determining the hypothesis \( h_j^* \) most supported by evidence from a set of possible hypotheses. Since the nodes have incomplete observations about SotW, they wish to collaborate with each other through a central server to share some information with each other to improve each node’s local distribution over states. However, due to communication constraints, every node can transmit only a limited amount of information to the server in each communication round. The central server \( CS \) aggregates messages from all nodes to update its own perception of the SotW, denoted by \( \ProbLoc{CS}{t} \). The server then selects a message and broadcasts it back to the nodes to aid them in refining their local distributions. In the next section, we turn our attention to, for each round t of communication: (i) how a node chooses its message for transmission to the server, (ii) how the server updates its perceived SotW upon receiving all incoming messages from the nodes, (iii) how the server selects which message to broadcast to all nodes, (iv) how a node updates it’s perceived SotW upon receiving the message from the server, and (v) how the hypothesis deduction problem is solved.

\subsection{Communication Protocol}

The communication occurs in iterative rounds, each consisting of two phases: an \textit{uplink phase} and a \textit{downlink phase}.

\paragraph{Uplink Phase} In round \( t \), each node \( N_j \) selects a message \( m^*_{j,t} \) (which is in the form of a FOL sentence) to transmit to the server \( CS \). The message \( m^*_{j,t} \) is chosen to maximize the semantic content informativeness with respect to the node's observations \( e_j \), while considering previously received messages \( \{m_{cs, z}\}_{z=1}^{t-1} \) from the server and previously transmitted messages to the server \( \{m_{j, z}\}_{z=1}^{t-1} \). The optimization objective for node \( N_j \) becomes:

\begin{equation}\label{cont1}
m_{j, t}^* = \argmin_{m \notin \bigcup_{z = 1}^{t-1} m_{cs, z} \land m_{j, z}} \text{cont}(e_j \land \bigcup_{z = 1}^{t-1} m_{cs, z}, m) \quad \text{s.t. } |m| \leq B,
\end{equation}

\noindent where $|m|$ is the number of FOL sentences message $m$ contains.

\paragraph{Downlink Phase} Upon receiving messages \( \{ m_{j, t} \}_{j=1}^{v} \) from all nodes, the server \( CS \) updates its distribution on state space \( \ProbLoc{CS}{t} \) based on the aggregated information:
\[
\ProbLoc{CS}{t} = \text{Update}(\ProbLoc{CS}{t-1}, \{ m^*_{j, t} \}_{k=1}^K),
\]
where \textit{Update} function involves Bayesian updating (c.f., eqn. \ref{c_func} and \ref{eq:posterior}) to incorporate new evidence and adjust probabilities over \( \mathcal{S} \).

The server then selects a message \( m_{CS, t} \) to broadcast, containing information not previously sent out:

\begin{equation}\label{cont2}
m_{CS, t}^* = \argmin_{m \notin \bigcup_{z = 1}^{t-1} m_{cs, z}} \text{ cont}(\bigcup_{z = 1}^{t} \{m^*_{j, z} \}_{j=1}^{V}, m) \quad \text{s.t. } |m| \leq B,
\end{equation}

\noindent while considering previously broadcasted messages from server \( \{m_{cs, z}\}_{z=1}^{t-1} \). Upon receiving \( m_{CS, t}^* \), node \( N_j \) updates its distribution:

\[
\ProbLoc{j}{t} = \text{Update}(\ProbLoc{j}{t-1}, m_{CS, t}^*).
\]

This iterative process allows nodes to progressively refine their understanding of the SotW. Once the communication concludes, nodes perform deduction.

\paragraph{Hypothesis Deduction}

After $T$ rounds of communication, each node $N_j$ choose the hypothesis best supported by evidence as:

\begin{equation}
    h_j^* = \argmax_{h} c(h, e_j \land \bigcup_{z = 1}^{T} m_{cs, z})
\end{equation}

In the next section, we turn our attention to the theoretical properties of this communication scheme.

\subsection{Convergence Analysis Under Finite Evidence}

In this section, we examine the convergence behavior of the communication framework under finite evidence to determine the relationship between empirical and true distribution over state space $\mathcal{S}$. To preserve generality, we assume an infinite universe $\WI$ and conduct the analysis over constituents. 

Let \( e_{j, l}^c \) describe a finite sample of \( l \) individual observations, containing \( c \) distinct attributive constituents accumulated at node $N_j$ after $T$ rounds of communication. A constituent \( C^w \) of width \( w \) is compatible with \( e_{j, l}^c \) only if \( c \leq w \leq K \), where \( K = 4^{|P| \times |E|^2} / 2\).

To compute the posterior distribution  

\begin{equation}\label{ProbLoc}
    \ProbLoc{j}{T} = \{c(C^w, e_{j, l}^c) \mid c \leq w \leq K\},
\end{equation}

one needs to specify the prior $c(C^w, \emptyset)$ and the likelihood $c(e_{j, l}^c, C^w)$. A natural method to assign the prior \( c(C^w, \emptyset) = P(C^w) \) is to assume it is proportional to \( \left( \frac{w}{K} \right)^\alpha \), reflecting the assumption that there are \( \alpha \) individuals compatible with a particular constituent \( C^w \). This approach captures the intuition that constituents with larger widths (i.e., involving more attributive constituents) are more probable if we believe there are more individuals exemplifying them. The explicit equation for the prior \( P(C^w) \) is:

\begin{equation} \label{eq:prior}
c(C^w, \emptyset) = P(C^w) = \frac{\Gamma\left( \alpha + \frac{w}{K} \right)}{\Gamma\left( \frac{w}{K} \right)} \bigg/ \sum_{i=0}^{K} \binom{K}{i} \frac{\Gamma\left( \alpha + \frac{i}{K} \right)}{\Gamma\left( \frac{i}{K} \right)},
\end{equation}
where \( \Gamma(\cdot) \) is the gamma function, and \( \alpha \geq 0 \) is a parameter reflecting our prior belief about the number of individuals compatible with \( C^w \).

The likelihood \( P(e_{j, l}^c \mid C^w) \) is computed based on the number \( l_z \) of observed samples confirming each attributive constituent \( C_z(x) \) in the evidence \( e_{j, l}^c \). It is given by:
\begin{equation} \label{eq:likelihood}
c(e_{j, l}^c, C^w) = P(e_{j, l}^c \mid C^w) = \frac{\Gamma(\lambda(w))}{\Gamma(l + \lambda(w))} \prod_{z=1}^{c} \frac{\Gamma\left( l_z + \frac{\lambda(w)}{w} \right)}{\Gamma\left( \frac{\lambda(w)}{w} \right)},
\end{equation}
where \( \lambda \) is a prior coefficient, possibly dependent on \( w \). This likelihood reflects how well constituent \( C^w \) explains the observed evidence \( e_{j, l}^c \).

Using the prior and likelihood, the posterior probability of constituent \( C^w \) given \( e_{j, l}^c \) is:
\begin{equation} \label{eq:posterior}
c(C^w, e_{j, l}^c) = P(C^w \mid e_{j, l}^c) = \frac{P(C^w) \, P(e_{j, l}^c \mid C^w)}{\sum_{w' \in \IndexSetConst} P(C^{w'}) \, P(e_{j, l}^c \mid C^{w'})},
\end{equation}
where the sum in the denominator runs over all constituents compatible with the language and $\IndexSetConst$ is the index set of constituents.

Using the prior in equation (\ref{eq:prior}) and the likelihood in equation (\ref{eq:likelihood}), we can analyze the convergence of the posterior probability \( P(C^w \mid e_{j, l}^c) \) as the number of observations \( l \) increases. Specifically, these equations are used to derive a Probably Approximately Correct (PAC) bound that determines the minimum number of samples required to consider the minimal constituent \( C^c \) to be approximately correct with high probability. (Despite omitted from this manuscript due to space concerns, it had been proven that asymptotically, the constituent with smallest width, a.k.a. the minimal constituent, receives probability 1 while all other constituents receive probability 0. See \cite{HINTIKKA19661} for further discussion and proof.)

\begin{theorem}[PAC-Bound for Constituents \cite{HINTIKKA19661}] \label{PAC_bound}
Given evidence \( e_{j, l}^c \) of size \( l \), let \( l_0 \) be such that \( P(C^c \mid e_l^c) \geq 1 - \epsilon \) for some \( l \geq l_0 \). Then:
\begin{equation} \label{err_bound}
\epsilon' \leq \max_{0 \leq c \leq K - 1} \left\{ \sum_{i=1}^{K - c} \binom{K - c}{i} \left( \frac{c}{c + i} \right)^{l - \alpha} \right\},
\end{equation}
where \( \epsilon' = \frac{\epsilon}{1 - \epsilon} \).
\end{theorem}

This PAC bound provides a theoretical guarantee on the convergence of the posterior distribution $\ProbLoc{j}{T}$ to true distribution on the state space $\mathcal{S}$ as more evidence is accumulated. The detailed derivation of this bound relies on the prior and likelihood expressions, as also discussed in the next section. Now, we turn our attention to showing that DISCD is superior compared to random transmissions.

\subsection{Effectiveness of Proposed Framework}

In this section, we compare two communication strategies under sentence level communication constraints, from the perspective of a node $N_j$, as detailed by the next theorem.

\begin{theorem}[Effectiveness of DISCD Framework]\label{thm:cont-informative-revised}
Given a fixed communication budget allowing the transmission of \( B \) FOL sentences over the entire course of communication (e.g., in $T$ rounds) by the central server, consider the following two strategies:

1. \emph{DISCD Message Selection}: Let $e_{j, cont}$ constitute the \( B \) sentences with the highest degree of confirmation that were selected and broadcast by central server during communication.

2. \emph{Random Message Selection}: Let $e_{j, rand}$ constitute \( B \) sentences selected uniformly at random and broadcast by central server during communication.

Assume the minimal constituent $C^c$ corresponds to the true state of the world. We establish the following theorem.

\begin{enumerate}[label=(\alph*)]
    \item The posterior probability of the minimal constituent satisfies
    \begin{equation}\label{eq:posterior-revised}
        c(C^c , e_{j, cont}) \geq c(C^c, e_{j, rand}).
    \end{equation}
    \item The evidence \( e_{j, cont} \) yields a more accurate posterior distribution $\ProbLoc{j, cont}{T} $ over constituents and yields a tighter PAC bound (as given by Theorem 3) compared to that of \( e_{j, rand} \), e.g., $\epsilon'_{cont} \leq \epsilon'_{rand}$.
\end{enumerate}
\end{theorem}

\begin{proof}

The posterior distributions $\ProbLoc{j, cont}{T}$ and $\ProbLoc{j, rand}{T}$ are given by (\ref{ProbLoc}) and (\ref{eq:posterior}). The likelihood \( P(e_{j, cont} \mid C^w) \) depends on how well constituent \( C^w \) explains the evidence \( e_{j, cont} \), and similarly for $P(e_{j, rand} \mid C^w)$. \( e_{j, cont} \) affects the likelihoods such that \( P(e_{j, cont} \mid C^c) \) is higher because \( C^c \) aligns closely with the specific observations in \( e_{j, cont} \). For incorrect constituents \( C^w \neq C^c \), \( P(e_{j, cont} \mid C^w) \) is lower due to poorer alignment with the detailed evidence as per (\ref{const_dis}). Therefore, the numerator \( P(C^w) P(e_{j, cont} \mid C^w) \) increases, while the denominators adjust based on the altered likelihoods, leading to:

\begin{equation}
    P(C^w \mid e_{j, cont}) \geq P(C^w \mid e_{j, rand}).
\end{equation}

Similarly, the posterior probabilities for incorrect constituents decrease under \( e_{j, cont} \), resulting in a more accurate posterior distribution. From the convergence analysis and Theorem 3, we can observe that the PAC bound depends on \( l \), the number of observations, and the likelihoods. As evidence \( e_{j, cont} \) increases with \( l \), it yields increasingly more discriminative likelihoods (i.e., higher probability concentration around more accurate constituents and the true constituent), tightening the PAC bound. Specifically, with higher \( l \) and more informative likelihoods, the sum in the PAC bound inequality decreases:

\begin{equation}\label{eq:pac_bound_revised}
    \epsilon' \leq \sum_{i=1}^{K - c} \binom{K - c}{i} \left( \frac{P(e \mid C^{c+i}) P(C^{c+i})}{P(e \mid C^c) P(C^c)} \right).
\end{equation}

With \( e_{j, cont} \), the ratio \( \frac{P(e_{j, cont} \mid C^{c+i})}{P(e_{j, cont} \mid C^c)} \) becomes smaller due to the more discriminative likelihoods, leading to a smaller \( \epsilon' \) (tighter error bound) compared to \( e_{j, rand} \).

\end{proof}

Elaborating on Theorem 4, it is observed that evidence \( e_{j, cont} \) affects the posterior in two ways. It provides more specific information and as the evidence is more informative (larger \( l \), larger \( l_w \)) and the likelihood function \( P(e_{j, cont} \mid C^w) \) will often be higher for constituents that align closely with the evidence and lower for others. Also, \( e_{j, cont} \) may render some constituents impossible (assigning them a likelihood of zero) because they are inconsistent with the evidence.

Under \( e_{j, rand} \), the likelihoods \( P(e_{j, rand} \mid C^w) \) are less discriminative because random evidence provides less specific information. The differences in likelihoods between the true constituent and incorrect ones are smaller. As the likelihoods under \( e_{j, cont} \) are more discriminative, the posterior distribution \( P(C^w \mid e_{j, cont}) \) is more peaked around the true constituent \( C^c \), assigning higher probability to \( C^c \) and lower probabilities to others.

In other words, \( e_{j, cont} \) improves discrimination. By providing more specific observations, evidence \( e_{j, cont} \) with high degree of confirmation enhances the differences in likelihoods between the true constituent \( C^c \) and incorrect constituents. The increased likelihood \( P(e_{j, cont} \mid C^c) \) and decreased likelihoods \( P(e_{j, cont} \mid C^w \neq C^c) \) result in a higher posterior probability for \( C^c \). The posterior distribution over constituents becomes more concentrated around \( C^c \), assigning lower probabilities to incorrect constituents and leading to a more accurate posterior. The improved likelihood ratios lead to a tighter PAC bound, reflecting a better estimation of \( C^c \).  In the next section, we comment on the impact of improved posterior estimation on hypothesis deduction task success.

\subsection{Analysis of Semantic Uncertainty and Task Success}

Albeit discussions in this section are for a single node $N_j$, they can be trivially extended to multi node case. Each node \( N_j \) aims to minimize its risk (expected loss) \( R_{\ProbLoc{j}{T} }(h_j) \) over \( \mathcal{B}_j \):

\[
R_{\ProbLoc{j}{T}}(h_j) = E_{h_j \sim \ProbLoc{j}{T} } [L(h_j, h_j^*)],
\]

where \( L(h_j, h_j^*) \) quantifies the discrepancy between \( h_j \) and \( h_j^* \) and can be selected as any appropriate distance metric.

As nodes exchange information and receive informative messages, the distributions \( \ProbLoc{j}{T} \) converge towards \(  \ProbLoc{j*}{T} \) (true distribution), reducing expected loss and enhancing their ability to identify \( h_j^* \), as $h_j \equiv \bigvee_{\substack{w \in \IndexSetConst \\ \text{such that } h_j \models C^w}} C^w$ and $c(h_j, e_j \land \bigcup_{z = 1}^{T} m_{cs, z}) = \sum_w c( C^w, e_j \land \bigcup_{z = 1}^{T} m_{cs, z})$. 

Therefore, as evidence accumulates and as per Theorems 3 \& 4, \( e_{j, cont} \) leads to tighter converge towards the true (minimal) constituent $C^c$, it better minimizes the Bayes risk \(  R_{\ProbLoc{j}{T}}(h_j) \) compared to random selection of messages \( e_{j, rand} \) (e.g., \(  R_{\ProbLoc{j, cont}{T}}(h_j) \leq  R_{\ProbLoc{j, rand}{T}}(h_j) \)), yielding larger hypothesis deduction success rate. In the next section, we experimentally verify these theoretical observations.

\section{Experiment Results}

The effectiveness of DISCD communication algorithm compared to random selection was tested by applying it to a custom deduction dataset. The relevant code and dataset are accessible on GitHub \cite{sih_far}. The custom deduction dataset describes a story where each node $N_j$ has only a portion of it, i.e., possess incomplete information about the story. The story is divided across all nodes, each node receiving equal number of sentences. In order to reflect possible overlaps in the worldviews of different nodes, 30\% of each node's data is shared with at least one other whereas the remaining 70\% is unique to that node. The task in each node is to determine, among 8 distinct hypothesis, the one most supported by the evidence (story) for the entities that are present in the node. However, as the information is incomplete, for some entities, nodes require information from others to be able to deduce the correct hypothesis.

During communication, each node identifies the most informative subset of \( B \) sentences from a total of \( N = 40 \) sentences by solving the optimization problem in (\ref{cont1}). Minimizing (\ref{cont1}) requires calculating inductive logical probabilities, which, in turn necessitates counting the number of states in $\mathcal{S}$ a FOL expression satisfies. As the state space $\mathcal{S}$ is exponentially large, we utilize a state-of-the-art algorithmic boolean model counting tool sharpSAT-td to determine the number of states. To shed more light on how the this optimization works, consider the following example. Let node $N_j$ has sentences $\{FOL_1, ..., FOL_{40}\}$ in its dataset. Assuming $ B=1 $, in round 0, the node calculates the number of states in $\mathcal{S}$ each sentence $FOL_i$ satisfies, and determines their inductive logical probability. Based on those probabilities, $N_j$ optimizes (\ref{cont1}). After transmitting, say $FOL_3$, and receiving another sentence from the server, call it $FOL_{cs, 0}$, it continues with the round 1, where it calculates the probabilities for $FOL_{cs, 0} \land FOL_3 \land FOL_i$ for $i \in {1, 2, 4, ..., 40}$. Then, another $FOL_i$ is selected, transmitted, and so on. For more details on how this optimization algorithm works, see \cite{saz2024, saz2024model}. In our experiments, we tested subsets with sizes \( B = 1, 2\). Once the communication concludes, the nodes perform their respective hypothesis deduction tasks, which is same for all nodes in our experiments. We measure the success rate with the ratio of the population (entities) where the correct hypothesis can be deduced. The results are presented in Fig. \ref{deduc}.

\begin{figure}
    \centering
    \includegraphics[width=\linewidth]{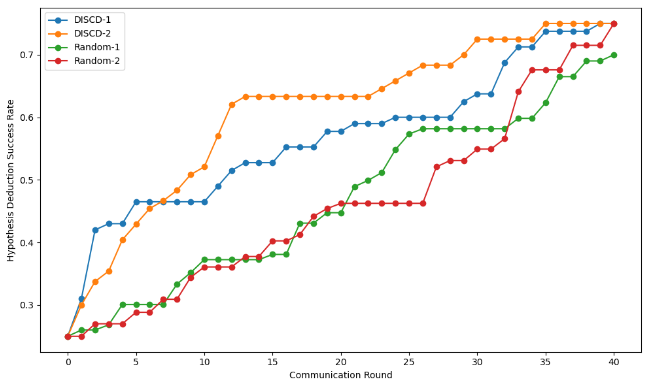}
    \caption{Success rate in hypothesis deduction task vs. communication round}
    \label{deduc}
\end{figure}

DISCD clearly outperforms random selection in hypothesis deduction performance, for both message sizes ($B = 1, 2$). In other terms, DISCD consistently selects sentences that significantly improve hypothesis deduction accuracy compared to random transmissions. Indeed, it only takes 5 rounds of communication for DISCD-1 to achieve 46.5\% success rate whereas it takes 22 rounds of communication for Random-1 to achieve the same performance. Similarly, it takes 13 rounds to achieve 62.1\% success rate for DISCD-2 whereas the same rate of success is achieved for Random-2 only after 34 rounds. After 40 rounds of communication, both DISCD-1 and DISCD-2 achieve 75.0\% performance whereas random-1 can only achieve 70.0\%. This demonstrates that DISCD can achieve similar performance in less number of rounds, saving from communication.

An important clarification shall be made regarding why the performance stalls from time to time. As each hypothesis has multiple premises which may require multiple pieces of essential information to be transmitted, it can take a number of rounds of communication until the performance improves. However, DISCD, thanks to rapid selection of important semantic information regarding SotW, improves accuracy rapidly before stalling, whereas in random selection, the stalling is distributed, and it takes longer to achieve the same level of performance.

Another clarification regards the distinction between DISCD-1 and DISCD-2. In DISCD-2, two FOL sentences are transmitted from nodes to server per round and server broadcasts two sentences back to nodes. This implies what sentences are selected and in which order can be different between the two cases and hence accounts for the different rates of increase in hypothesis deduction performance.  

As the communication machinery, i.e., “compression and channel coding” is common between random method and our method, we did not consider them in comparative results. However, based on the number of sentences transmitted, we studied the average communication costs per node for both schemes for different levels of accuracies in deduction. The results are presented at Table \ref{mytable}.

\begin{table}[htbp]
\centering
\begin{tabular}{|c|c|c|c|c|}
\hline
\textbf{Success Rate} & \textbf{DISCD-1} & \textbf{Random-1} & \textbf{DISCD-2} & \textbf{Random-2} \\ \hline
40\% & 70.9 bits & 402.3 bits & 94.6 bits & 354.9 bits \\ \hline
45\% & 141.9 bits & 496.9 bits & 141.9 bits & 449.6 bits \\ \hline
50\% & 307.6 bits & 544.3 bits & 212.9 bits & 638.9 bits \\ \hline
55\% & 402.3 bits & 591.6 bits & 260.3 bits & 757.3 bits \\ \hline
60\% & 473.3 bits & 828.3 bits & 283.9 bits & 780.9 bits \\ \hline
65\% & 591.6 bits & 851.9 bits & 567.9 bits & 804.6 bits \\ \hline
70\% & 780.9 bits & 946.6 bits & 686.3 bits & 875.6 bits \\ \hline
75\% & 946.6 bits & N/A & 828.3 bits & 946.6 bits \\ \hline
\end{tabular}
\caption{Success Rate vs. Communication Cost (uplink, per node (average))}
\label{mytable}
\end{table}

Based on Table \ref{mytable}, DISCD-1 method, on average, saves approximately \textbf{270.49} bits of communication per node compared to random selection Random-1. Similarly, DISCD-2 method, on average, saves approximately \textbf{316.56} bits of communication compared to random selection Random-2. 

In conclusion, the experiments demonstrate that the proposed semantic communication algorithm DISCD effectively selects and transmits the most informative sentences, achieving substantial success for deductive hypothesis selection task with fewer bits compared to random selection faster. Furthermore, these results validates the theoretical findings regarding convergence, effectiveness and task success. DISCD framework better informs the nodes regarding SotW, leading to better task success.

\bibliographystyle{unsrt} 
\bibliography{refs} 

\end{document}